\documentclass[10pt,conference]{IEEEtran}
\makeatletter
\def\ps@headings{%
\def\@oddhead{\mbox{}\scriptsize\rightmark \hfil \thepage}%
\def\@evenhead{\scriptsize\thepage \hfil \leftmark\mbox{}}%
\def\@oddfoot{}%
\def\@evenfoot{}}
\makeatother
\usepackage{amssymb,color,amsmath}
\usepackage{paralist,cite}
\usepackage{graphicx}

\newcommand{\sdual}{{\perp_s}}

\newcommand{\hdual}{{\perp_h}}
\newcommand{\F}{\mathbf{F}}

\DeclareMathOperator{\swt}{swt} \DeclareMathOperator{\wt}{wt}
 
\newtheorem{theorem}{\textbf{Theorem}}
\newtheorem{lemma}[theorem]{\textbf{Lemma}}
\newtheorem{definition}[theorem]{\textbf{Definition}}
\newtheorem{exampleX}[theorem]{\textbf{Example}}

\newcommand{\nix}[1]{}

\begin{document}
\title{Nonbinary Quantum Cyclic and Subsystem Codes Over Asymmetrically-decohered Quantum Channels}
\author{
\authorblockN{Salah A. Aly}
\authorblockA{Department of Electrical Engineering\\ Princeton University,  NJ, USA \\ salah@princeton.edu} \and
\authorblockN{Alexei Ashikhmin}
\authorblockA{ Bell Labs \&
Alcatel Lucent  \\ Murray Hill, NJ, USA \\ aea@alcatel-lucent.com}

 } \maketitle

\begin{abstract}
Quantum computers theoretically are able to solve
certain problems more quickly than any deterministic or
probabilistic computers. A quantum computer exploits the rules of quantum mechanics to speed up computations. However, one has to mitigate the resulting noise and decoherence effects to
avoid computational errors in order to successfully build quantum  computers.

 In this paper, we construct  asymmetric quantum codes to protect quantum information over asymmetric quantum channels, $\Pr Z \geq \Pr X$. Two generic methods are presented to derive asymmetric quantum cyclic codes  using the generator
polynomials and defining sets of classical cyclic codes. Consequently, the methods allow us to construct several families of quantum BCH, RS, and RM codes over asymmetric quantum channels. Finally, the methods are used to construct families of asymmetric subsystem codes.
\end{abstract}

\section{Introduction}
Quantum computers theoretically are able to solve
certain problems more quickly than any deterministic or
probabilistic computers. An example of such problems is the factorization of large integers in polynomial time. The novel idea is that a quantum computer exploits the rules of quantum mechanics to speed up computations. However, one has to mitigate the resulting noise and decoherence effects to
avoid computational errors in order to successfully build quantum  computers.
Recently, the theory of quantum error-correcting codes is extended to include construction of such codes over asymmetric quantum channels --- qubit-flip and phase-shift errors may have equal or
different probabilities, $\Pr Z \geq \Pr X$, the terminology is explained later. Asymmetric quantum error control codes (AQEC) are quantum codes
defined over biased quantum channels. Construction of such codes first
appeared in~\cite{evans07,ioffe07,stephens07}. The code construction of AQEC is the CSS construction of QEC
based on two classical cyclic codes. For more details on the CSS
constructions of QEC see for
example~\cite{shor95,ashikhmin01,steane96,steane96b,steane97,calderbank98}

There have been several attempts to characterize the noise error
model in quantum information~\cite{chang00}. In~\cite{steane96} the
CSS construction of a quantum code that corrects the errors
separated was stated. However, the percentage between the qubit-flip
and phase-shift error probabilities was not known for certain
physical realization. Recently, quantum error correction has been
extended over amplitude-damping channels~\cite{fletcher07}.

We expand the construction of  quantum error correction by
designing stabilizer codes that can correct phase-flip and
qubit-flip errors separately.  Assume that the quantum noise
operators occur independently and with different probabilities in
quantum states. Our goal is to adapt the constructed quantum codes
to more realistic noise models based on an appropriate physical phenomena.

Motivated by their classical counterparts, the asymmetric quantum cyclic codes
that we derive have online simple encoding and decoding circuits
that can be implemented using shift-registers with feedback
connections. Also, their algebraic structure makes it easy to derive
their code parameters. Furthermore, their stabilizer can be defined
easily using generator polynomials of classical cyclic codes, in addition, it
is simple to derive self-orthogonal nested-code conditions for these
cyclic classes of codes.

In this paper we construct quantum error-correcting codes that correct
quantum errors that may destroy quantum information with different
probabilities. We derive two generic framework methods that can be
applied to any classical cyclic codes in order to derive asymmetric quantum
cyclic codes. The methods are used to derive Asymmetric quantum BCH, RM, RS  codes. In addition, they are used to derive families of asymmetric subsystem codes over finite fields. Several classes of asymmetric quantum codes  are also shown
in~\cite{aly08aqec,ioffe07,sarvepalli08}.

\bigskip

\emph{Notation:} Let $q$ be a power of a prime integer $p$. We
denote by $\F_q$ the finite field with $q$ elements. We define the
Euclidean inner product $\langle x|y\rangle =\sum_{i=1}^nx_iy_i$ and
the Euclidean dual of a code $C\subseteq \F_{q}^n$ as $$C^\perp = \{x\in
\F_{q}^n\mid \langle x|y \rangle=0 \mbox{ for all } y\in C \}.$$ We
also define the Hermitian inner product for vectors $x,y$ in
$\F_{q^2}^n$ as $\langle x|y\rangle_h =\sum_{i=1}^nx_i^qy_i$ and the
Hermitian dual of $C\subseteq \F_{q^2}^n$ as
$$C^\hdual= \{x\in \F_{q^2}^n\mid \langle x|y \rangle_h=0 \mbox{ for all } y\in
C \}.$$
An $[n,k,d]_q$ denotes a classical code $C$ with  length $n$, dimension $k$, and minimum distance $d$ over $\F_q$.  A quantum code $Q$ is denoted by $[[n,k,d]]_q$.

\bigskip

\section{Classical Cyclic Codes}
Cyclic
codes are of greater interest because they have efficient encoding
and decoding algorithms. In addition, they have well-studied
algebraic structure. Let $n$ be a positive integer and $\F_q$ be a
finite field with $q$ elements. A cyclic code $C$ is a principle
ideal of

$$R_n=\F_q[x] / (x^n-1),$$

where $\F_q[x]$ is the ring of polynomials in invariant $x$. Every
cyclic code $C$ is generated by either a generator polynomial $g(x)$
or generator matrix $G$. Furthermore, every cyclic code is a linear
code that has dimension $k=n-deg(g(x))$. Let $c(x)$ be a codeword in
$\F_q^n[x]$ then $c(x)=m(x)g(x)$, where $m(x)$ is the message to be
encoded. Consequently, every codeword can be written uniquely using
a polynomial in $\F_q^n[x]$. Also, a codeword $c$ in $C$ can be
written as $(c_0,c_1,...,c_{n-1}) \in \F_q^n$. A codeword $c(x) \in
\F_q^n[x]$ is in $C$ with defining set $T$ if and only if
$c(\alpha^i)=0$ for all $i \in T$. Every cyclic code generated by a
generator polynomial $g(x)$ has a parity check polynomial
$x^{k}h(1/x)/h(0)$ where $h(x)=(x^n-1)/g(x)$. Clearly, the parity
check polynomial $h(x)$ can be used to define the dual code
$C^\perp$ such that $g(x)h(x) \mod (x^n-1)=0$. Recall that the dual
cyclic code $C^\perp$ is defined by the generator polynomial
$g^\perp(x)=x^{k}h(x^{-1})/h(0)$. Let $\alpha$ be an element in
$\F_q$. Then sometimes, the code is defined by the roots of the
generator polynomial $g(x)$. Let $T$ be the set of roots of $g(x)$,
$T$ is the defining set of $C$, then
$$
g(x) = \prod_{i \in T}(x-\alpha^i).$$ The set $T$ is the union of
cyclotomic cosets modulo $n$ that has $\alpha^i$ as a root.  More
details in cyclic codes can be found
in~\cite{huffman03,macwilliams77}. The following Lemma is needed to
derive cyclic AQEC.

\medskip

\begin{lemma}\label{lem:twocycliccodes}
Let $C_i$ be cyclic codes of length $n$ over $\F_q$ with defining
set $T_i$ for $i=1,2$. Then \begin{compactenum}[i)] \item $C_1 \cap
C_2$ has defining set $T_1 \cup T_2$. \item $C_1+C_2$ has defining
set $T_1 \cap T_2$. \item $C_1 \subseteq C_2$ if and only if $T_2
\subseteq T_1$.
\item $C_i^\perp \subseteq C_{1 + i(\mod 2)}$ if and only if  $C_{1 + i(\mod 2)}^\perp  \subseteq C_i
$.
\end{compactenum}
\end{lemma}

\bigskip

We will provide an analytical method not a computer search method to derive such codes. The
benefit of this method is that it is much easier to derive families
of AQEC. We define the classical cyclic code using the defining
set and generator polynomial~\cite{aly07a}, \cite{huffman03}. The following lemma establishes
conditions when $C_2^\perp \subseteq C_1$.

\medskip

\begin{lemma}
Let $T_{C_i}$ and $g_i(x)$ be the defining set and generator
polynomial of a cyclic code $C_i$ for $i=\{1,2\}$. If one of the
following conditions

\begin{compactenum}[i)]
\item $T_{C_1} \subseteq T_{C_2}$,
\item $g_1(x)$ divides $g_2(x)$,
\item $h_2(x)$ divides $h_1(x)$,
\end{compactenum}
then  $C_2 \subseteq C_1$.
\end{lemma}
\begin{proof}
The proof is straight forward from the definition of the codes $C_1$
and $C_2$ and by using Lemma~\ref{lem:twocycliccodes}.
\end{proof}

\bigskip

\section{Deriving Asymmetric Quantum Codes}

We will show how to derive asymmetric quantum cyclic codes based on
a given classical cyclic code using the CSS construction as follows.

Let $H_i$ and $G_i$ be the parity check and generator  matrices of a
classical code $C_i$
 with parameters $[n,k_i,d_i]_2$ for $i \in \{1,2\}$. The
 commutativity condition of $H_1$ and $H_2$ is stated as

\begin{eqnarray}
H_1.H_2^T+H_2.H_1^T=\textbf{0}.
 \end{eqnarray}

Without loss of generality, we will assume that one of these two
classical codes controls the phase-shift errors, while the other
codes controls the bit-flip errors. Hence the CSS construction of a
binary AQEC can be stated as follows. Hence the codes $C_1$ and
$C_2$ are mapped to $H_x$ and $H_z$, respectively.

\medskip

\begin{definition}Given two classical binary codes $C_1$ and $C_2$ such that $C_2^\perp
\subseteq C_1$. If we form $ G=\begin{pmatrix}
G_1&0\\0&G_2\end{pmatrix}, \mbox{  and   } H =\begin{pmatrix}
H_1&0\\0&H_2\end{pmatrix}, $ then
\begin{eqnarray}
H_1.H_2^T-H_2.H_1^T=0
\end{eqnarray}
 Let $d_1=\min\{\wt(C_1\backslash C_2^\perp),wt(C_2\backslash
C_1^\perp) \}$ and $d_2=\max \{wt(C_2\backslash
C_1^\perp), \wt(C_1\backslash C_2^\perp)\}$, such that $k_1+k_2>n$. If we assume that
 $C_1$ corrects the qubit-flip errors and $C_2$ corrects the phase-shift errors, then there exists
AQEC with parameters
\begin{eqnarray}
[[n,k_1+k_2-n,d_2/d_1]]_2.
\end{eqnarray}\end{definition}

\medskip

The following theorem shows the CSS construction of asymmetric
quantum error control codes over $\F_q$.

\medskip

\begin{theorem}[CSS AQEC]\label{lem:AQEC}
Let $C_1$ and $C_2$ be two classical codes with parameters
$[n,k_1,d_1]_q$ and $[n,k_2,d_2]_q$ respectively, and $d_x=
\min\big\{\wt(C_{1} \backslash C_2^\perp), \wt(C_{2} \backslash C_{1
}^\perp)\big\}$, and $d_z= \max\big\{\wt(C_{1} \backslash
C_2^\perp), \wt(C_{2} \backslash C_{1 }^\perp)\big\}$. If
  $C_2^\perp \subseteq C_1$, then
\begin{compactenum}[i)]
\item  there exists an AQEC with parameters $[[n,\dim C_1 -\dim
C_2^\perp,d_z/d_x]]_q$ that is $[[n,k_1+k_2-n,d_z/d_x]]_q$.
 Also, there
exists a QEC with parameters $[[n,k_1+k_2-n,d_x]]_q$.
\item there exists an asymmetric subsystem code with parameters $[[n,k_1+k_2-n-r,r,d_z/d_x]]_q$ for $0 \leq r \leq k_1+k_2-n$.
\end{compactenum} Furthermore, all constructed codes are pure to their minimum
distances.
\end{theorem}

Therefore, it is straightforward to derive asymmetric quantum
control codes from two classical codes as shown in
Lemma~\ref{lem:AQEC} as well as a subsystem code. Of course, one wishes to increase the values
of $d_z$ vers. $d_x$ for the same code length and dimension.
\bigskip

If the AQEC has minimum distances $d_z$ and $d_x$ with $d_z \geq d_x$, then it can
correct all qubit-flip errors $ \leq \lfloor (d_x-1)/2\rfloor$ and
all phase-shift errors $ \leq \lfloor (d_z-1)/2\rfloor$,
respectively, as shown in the following result.

\medskip

\begin{lemma}
An $[[n,k,d_z/d_x]]_q$ asymmetric quantum code corrects all
qubit-flip errors up to $\lfloor (d_x-1)/2\rfloor$ and all
phase-shift errors up to $\lfloor (d_z-1)/2 \rfloor$.
\end{lemma}
The codes derived in~\cite{aly07a} for primitive and
nonprimitive quantum BCH codes assume that qubit-flip errors,
phase-shift errors, and their combination occur with equal
probability, where $\Pr{Z}=\Pr{X}=\Pr{Y}=p/3$, $\Pr{I}=1-p$, and
$\{X,Z,Y,I\}$ are the binary Pauli operators $P$, see~\cite{calderbank98,shor95}. We aim to
generalize these quantum BCH codes over asymmetric quantum channels. Furthermore, we will derive a much larger class of AQEC based on any two cyclic codes. Such codes include RS, RM, and Hamming codes.

\bigskip

\section{Asymmetric Quantum Cyclic Codes}\label{sec:aqcc}
Recently the theory of quantum error-correcting codes (QEC) has been
extended  to asymmetric quantum error-correcting codes (AQEC), in which the quantum
errors has biased probabilities. In this section we will give two
methods to derive asymmetric quantum cyclic codes. One method is
based on the generator polynomial of a cyclic code, while the other
is directly from the defining set of cyclic code.

\subsection{AQEC Based on Generator Polynomials of Cyclic
Codes}  Let $C_1$ be a cyclic code with parameters
$[[n,k,d]]_q$ defined by a generator polynomial $g_1(x)$. Let
$S=\{1,2,\ldots,\delta_1-1\}$, for some integer $\delta_1 <n$, be the set of roots of the polynomial
$g_1(x)$ such that
\begin{eqnarray}g_1(x)=\prod_{i \in S}(x-\alpha^i)\end{eqnarray}

 It is a well-known fact that the dimension of
the code $C_1$ is given by
\begin{eqnarray}
k_1=n-\deg(g_1(x))
\end{eqnarray}
We also know that the dimension of the dual code $C_1^\perp$ is
given by $k_1^\perp=n-k_1=\deg(g_1(x))$.

The idea that we propose is  simple. Let $f(x)=(x^b-1)$ be a
polynomial such that $1 \leq \deg (f(x)) \leq n-k$. We extend the
polynomial $g_1(x)$ to the polynomial $g_2^\perp(x)$ such that

\begin{eqnarray}
g_2^\perp(x)=f(x) g_1(x)
\end{eqnarray}

Now, let $g_2^\perp(x)$ be the generator polynomial of the code
$C_2^\perp$ that has dimension $k_2^\perp=n-deg(f(x)g_1(x)) <k_1$.
From the cyclic structure of the codes $C_1$ and $C_2^\perp$, we can
see that $C_2^\perp<C_1$, therefore $C_1^\perp < C_2$. Let
$d_1=\wt(C_1 \backslash C_2^\perp)$ and $d_2=\wt(C_2 \backslash
C_1^\perp)$ then we have the following theorem. We can also change
the rules of the code $C_1$ and $C_2$ to make sure that $d_2 >d_1$.

\medskip

\begin{theorem}
Let $C_1$ be a cyclic code with parameters $[n,k_1,d_1]_q$ and a
generator polynomial $g_1(x)$. Let $C_2^\perp$ be a cyclic code
defined by the polynomial $f(x)g_1(x)$ such that $b= \deg(f(x)) \geq
1$, then there exists AQEC with parameters
$[[n,2k_1-b-n,d_z/d_x]]_q$, where $d_x=\min \{\wt(C_1 \backslash
C_2^\perp), \wt(C_2 \backslash C_1^\perp)\}$ and $d_z=\max \{\wt(C_1
\backslash C_2^\perp), \wt(C_2 \backslash C_1^\perp)\}$. Furthermore
the code can correct $\lfloor (d_x-1)/2 \rfloor$ qubit-flip errors
and  $\lfloor (d_z-1)/2 \rfloor$ phase-shift errors.
\end{theorem}
\begin{proof}We proceed the proof as follows.
\begin{compactenum}[i)]
\item
We know that the dual code $C_1^\perp$ has dimension
$k_1^\perp=\deg(g_1(x))$. Also, $C_1^\perp$ has a generator
polynomial $h_1(x)=x^{n-k} h_1'(1/x)$ where
$h_1'(x)=(x^n-1)/g_1(x)$. Let $f(x)$ be a nonzero polynomial such
that $f(x)g_1(x)$ defines  a code $C_2^\perp$. Now the code
$C_2^\perp$ has dimension
$k_2^\perp=n-\deg(f(x)g_1(x))=n-(k_1+b)<k_1$.
\item
We notice that the polynomial $g_1(x)$ is a factor of the polynomial
$f(x)g_1(x)$, therefore the code generated by later is a subcode of
the code generated by the former. Then we have $C_2^\perp \subset
C_1$. Hence, the code $C_2^\perp$ has dimension
$k_2^\perp=n-(k_1+b)$.
\item
Also, the code $C_2$ has dimension $k_1+b$ and generator polynomial
given by $g_2(x)=(x^n-1)/(f(x)g_1(x))=h_1(x)/f(x)$. Hence the
$g_2(x)$ is a factor of $h_1(x)$, therefore $C_1^\perp$ is  a
subcode in $C_2$, $C_1^\perp \subseteq C_2$. There exists asymmetric
quantum cyclic code with parameters

\begin{compactenum}
\item
$\dim C_1 - \dim C_2^\perp=k_1-(n-k_1-b).$
\item $d_x=\min \{ \wt(C_2\backslash C_1^\perp),\wt(C_1\backslash
C_2^\perp)\}$ and  $d_z=\max \{ \wt(C_2\backslash
C_1^\perp),\wt(C_1\backslash C_2^\perp)\}$.
\end{compactenum}
\end{compactenum}
\end{proof}

\bigskip

\begin{figure}[t]
  \begin{center}
  \includegraphics[scale=0.65]{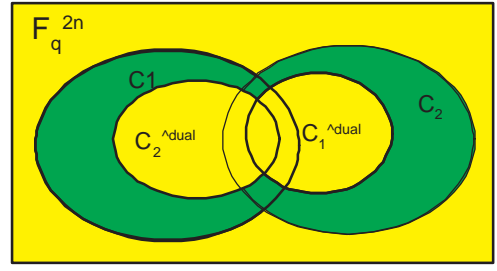}
  \caption{Constructions of asymmetric quantum codes (AQECs) based on two classical codes $C_1$ and $C_2$
  with parameters $[n,k_1]$ and $[n,d_2]$ such that $C_i \subseteq C_{1+(i \mod 2)}$ for $i=\{1,2\}$. AQEC has parameters $[[n,k_1+k_2-n,d_z/d_x]]_q$ where
  $d_x=\min\{\wt(C_1 \backslash C_2^\perp),\wt(C_2 \backslash C_1^\perp)\}$ and $d_z=\max\{\wt(C_2 \backslash C_1^\perp),\wt(C_1 \backslash C_2^\perp)\}$.}\label{fig:subsys1}
  \end{center}
\end{figure}

\subsection{Cyclic AQEC using the Defining Sets Extension}
We can  give a general construction for a cyclic AQEC over $\F_q$ if
the defining sets of the classical cyclic codes are known.

\medskip

\begin{theorem}\label{lem:cyclic-subsysI}
Let $C_1$ be a $k$-dimensional cyclic code of length $n$ over
$\F_q$. Let $T_{C_1}$ and $T_{C_1^\perp}$ respectively denote the
defining sets of $C_1$ and $C_1^\perp$. If $T$ is a subset of
$T_{C_1^\perp} \setminus T_{C_1}$ that is the union of cyclotomic
cosets, then one can define a cyclic code $C_2$ of length $n$ over
$\F_q$ by the defining set $T_{C_2}= T_{C_1^\perp} \setminus (T \cup
T^{-1})$. If $b=|T\cup T^{-1}|$ is in the range $0\le b< 2k-n$ then
there exists asymmetric quantum code with parameters
$$[[n,2k-b-n,d_z/d_x]]_q,$$ where $d_x= \min \{ \wt(C_2 \setminus C_1^\perp),\wt(C_1 \setminus C_2^\perp) \}$
and $d_z= \max \{ \wt(C_2 \setminus C_1^\perp),\wt(C_1 \setminus
C_2^\perp) \}$.
\end{theorem}
\begin{proof}
Observe that if $s$ is an element of the set  $S=
T_{C_1^\perp}\setminus T_{C_1} = T_{C_1^\perp} \setminus (N\setminus
T_{C_1^\perp}^{-1})$, then $-s$ is an element of $S$ as well. In
particular, $T^{-1}$ is a subset of $T_{C_1^\perp}\setminus
T_{C_1}$.

By definition, the cyclic code $C_2$ has the defining set $T_{C_2}=
T_{C_1^\perp} \setminus (T \cup T^{-1})$; thus, the dual code
$C_2^\perp$ has the defining set
$$T_{C_2^\perp}=N\setminus T_{C_2}^{-1} =
T_{C_1}\cup (T\cup T^{-1}).$$

Since $n-k=|T_{C_1}|$ and $b=|T\cup T^{-1}|$, we have $\dim_{\F_q}
C_1=n-|T_{C_1}| = k$ and $\dim_{\F_q} C_2 = n-|T_{C_2}|=k+b$. Thus,
 there exists an
$\F_q$-linear asymmetric quantum code $_Q$ with parameters
$[[n,k_Q,d_z/d_x]]_q$, where
\begin{compactenum}[i)]
\item $k_Q = \dim C_1 -\dim C_2^\perp=k-(n-(k+b))=2k+b-n$,
\item $d_x= \min \{ \wt(C_2 \setminus C_1^\perp),\wt(C_1 \setminus C_2^\perp) \}$
and $d_z= \max \{ \wt(C_2 \setminus C_1^\perp),\wt(C_1 \setminus
C_2^\perp) \}$.
\end{compactenum}
as claimed.
\end{proof}

The usefulness of the previous theorem is that one can directly
derive asymmetric quantum codes from the set of roots (defining set)
of a cyclic code. We also notice that the integer $b$ represents a
size of a cyclotomic coset (set of roots), in other words, it does
not represent one root in $T_{C_1^\perp}$.

\begin{table}[t]
\caption{Families of asymmetric quantum Cyclic codes}
\label{table:bchtable}
\begin{center}
\begin{tabular}{|c|c|c|c|c|}
\hline   q & $C_1$ BCH Code & $C_2$ BCH Code &AQEC \\
 \hline
 &&&\\
 2&$[15,11,3]$&$[15,7,5]$&$[[15,3,5/3]]_2$\\
 2&$[15,8,4]$&$[15,7,5]$&$[[15,0,5/4]]_2$\\
 2&$[31, 21, 5]$ & $[31, 16, 7]$& $[[31,6, 7/5]]_2$\\
 2&$[31,26,3]$&$[31,16,7]$&$[[31,11,7/3]]$\\
 2&$[31,26,3]$&$[31,16,7]$&$[[31,10,8/3]]$\\
 2&$[31,26,3]$&$[31,11,11]$&$[[31,6,11/3]]$\\
  2&$[31,26,3]$&$[31,6,15]$&$[[31,1,15/3]]$\\
2&$[127,113,5]$&$[127,78,15]$&$[[127,64,15/5]]$\\
2&$[127,106,7]$&$[127,77,27]$&$[[127,56,25/7]]$\\

  \hline
\end{tabular}
\end{center}
\end{table}
\bigskip

\section{AQEC and Connection with Subsystem
Codes}\label{sec:AQEC-subsystem}

In this section we establish the connection between AQEC and
subsystem codes. Furthermore we derive a larger class of quantum
codes called asymmetric subsystem codes (ASSC). We derive families
of subsystem BCH codes and cyclic subsystem codes over $\F_q$.
In~\cite{aly08a} we construct several families of subsystem
cyclic, BCH, RS and MDS codes over $\F_{q^2}$ with much more details

We expand our understanding of the theory of quantum error control
codes by correcting the quantum errors $X$ and $Z$ separately using
two different classical codes, in addition to correcting only errors
in a small subspace. Subsystem codes are a generalization of the
theory of quantum error control codes, in which errors can be
corrected as well as avoided (isolated).

Let $Q$ be a quantum code  such that $\mathcal{H}=Q\oplus Q^\perp$,
where $Q^\perp$ is the orthogonal complement of $Q$. We can define
the subsystem code $Q=A\otimes B$, see Fig.\ref{fig:subsys1}, as
follows

\medskip

\begin{definition}[Subsystem Codes]
An $[[n,k,r,d]]_q$ subsystem code is a decomposition of the subspace
$Q$ into a tensor product of two vector spaces A and B such that
$Q=A\otimes B$, where $\dim A=q^k$ and $\dim B= q^r$. The code $Q$
is able to detect all errors  of weight less than $d$ on subsystem
$A$.
\end{definition}
Subsystem codes can be constructed  from the classical codes  over
$\F_q$ and $\F_{q^2}$. Such codes do not need the classical codes to
be self-orthogonal (or dual-containing) as shown in the Euclidean
construction. We have given  general constructions of subsystem
codes in~\cite{aly06c} known as the subsystem CSS and Hermitian
Constructions. We provide a proof for the following special case of
the CSS construction.

\medskip

\begin{theorem}[ASSC Euclidean Construction]\label{lem:css-Euclidean-subsys}
If $C_1$ is a $k_1$-dimensional $\F_q$-linear code of length $n$ that
has a $k_2$-dimensional subcode $C_2=C_1\cap C_1^\perp$ and
$k_1+k_2<n$, then there exist
\begin{eqnarray}
[[n,n-(k_1+k_2),k_1-k_2,d_z/d_x]]_q,\\  \nonumber  [[n,k_1-k_2,n-(k_1+k_2),d_z/d_x]]_q \end{eqnarray}
subsystem codes, where $d_z=\max\{\wt(C_2^\perp\setminus C_1),\wt(C_1^\perp\setminus C_2)\}$ and $d_x=\min\{\wt(C_2^\perp\setminus C_1),\wt(C_1^\perp\setminus C_2)\}$.
\end{theorem}
\begin{proof}
The proof can be proceeded by defining pairs of codes as follows.
Let us define the code $X=C_1\times C_1 \subseteq \F_q^{2n}$,
therefore $X^\sdual=(C_1\times C_1)^\sdual=C_1^\sdual\times
C_1^\sdual$. Hence $Y=X \cap X^\sdual=(C_1\times C_1)\cap
(C_1^\sdual\times C_1^\sdual)= C_2 \times C_2$. Thus, $\dim_{\F_q}
Y=2k_2$. Hence $|X||Y|=q^{2(k_1+k_2)}$ and $|X|/|Y|=q^{2(k'-k'')}$.
By Theorem~\cite[Theorem 1]{aly06c}, there exists a subsystem code
$Q=A\otimes B$ with parameters $[[n,\log_q\dim A,\log_q\dim B,
d_z/d_x]]_q$ such that
\begin{compactenum}[i)]
\item $\dim A=q^n/(|X||Y|)^{1/2}=q^{n-k_1-k_2}$.
\item $\dim B=(|X|/|Y|)^{1/2}=q^{k_1-k_2}$.
\item $d_z= \max\{\swt(Y^\sdual \backslash X), \swt(X^\sdual \backslash Y) \}=\max\{\wt(C_2^\perp\setminus C_1),\wt(C_1^\perp\setminus C_2) \}$, and $d_x= \min\{\swt(Y^\sdual \backslash X), \swt(X^\sdual \backslash Y) \}=\min\{\wt(C_2^\perp\setminus C_1),\wt(C_1^\perp\setminus C_2) \}$
\end{compactenum}
Exchanging the rules of the codes $C_1$ and $C_1^\perp$ gives us the
other subsystem code with the given parameters.
\end{proof}

Subsystem codes (SCC) require the code $C_2$ to be self-orthogonal,
$C_2 \subseteq C_2^\perp$. AQEC and SSC are both can be constructed
from the pair-nested classical codes, as we call them. From this
result, we can see that any two classical codes $C_1$ and $C_2$ such
that $C_2=C_1 \cap C_1^\perp \subseteq C_2^\perp$, in which they can
be used to construct a subsystem code (SSC), can be also used to
construct asymmetric quantum code (AQEC). Asymmetric subsystem codes
(ASSC) are much larger class than  the class of symmetric subsystem
codes, in which the quantum errors occur with different
probabilities in the former one and have equal probabilities in the
later one. In short, AQEC does  not require the intersection
code to be self-orthogonal.

The construction in Lemma~\ref{lem:css-Euclidean-subsys} can be
generalized to ASSC CSS construction in a similar way. This means
that we can look at an AQEC with parameters $[[n,k,d_z/d_x]]_q$. as
subsystem code with parameters $[[n,k,0,d_z/d_x]]_q$. Therefore all
results shown in~\cite{aly08a,aly06c} are a direct
consequence by just fixing the minimum distance condition.

We have shown in~\cite{aly08a} that All stabilizer codes
(pure and impure) can be reduced to subsystem codes as shown in the
following result.

\medskip

\begin{theorem}[Trading Dimensions of ASSC and Co-SCC]\label{th:FqshrinkK}
Let $q$ be a power of a prime~$p$. If there exists an $\F_q$-linear
$[[n,k,r,d_z/d_x]]_q$ asymmetric subsystem code (stabilizer code if $r=0$) with $k>1$
that is pure to $d'$, then there exists an $\F_q$-linear
$[[n,k-1,r+1,\geq d_z/d_x]]_q$ subsystem code that is pure to
$\min\{d_x,d'\}$.  If a pure ($\F_q$-linear) $[[n,k,r,d_z/d_x]]_q$ asymmetric subsystem
code exists, then a pure ($\F_q$-linear) $[[n,k+r,d_z/d_x]]_q$ stabilizer
code exists.
\end{theorem}

\bigskip

\section{AQEC based on Two Cyclic Codes}

In this section we can also derive asymmetric quantum codes based on
two cyclic codes and their intersections. We do not necessarily
assume that the code $C_1$ is an extension of the code $C_2^\perp$.
However, we assume that $C_2^\perp \subset C_1$. The benefit of
designing AQEC based on two different classical codes is that we
guarantee the minimum distance $d_z$ to be large in comparison to
$d_x$. In this case we can assume that $C_1$ is a binary BCH code
with small minimum distance, while $C_2$ is an LDPC code with large
minimum distance.

The only requirement one needs to satisfy is that $C_{i} \subseteq
C_{1+i(\mod 2)}$. There have been many families that satisfy this
condition. For example $(15,7)$ BCH code turns out to be an LDPC
code. We will show an example to illustrate our theory.

\subsection{Illustrative Examples}
The following  example  illustrates the previous constructions. It gives a family of asymmetric quantum codes derived from the Hamming code with fixed minimum distance, and a BCH code with various designed distance.

\bigskip

\begin{exampleX}
Let $C_1$ be the Hamming code with parameters $[n,k,3]_2$ where
$n-2^m-1$ and $k=2^m-m-1$. Consider $C_2$ be a BCH code with
parameters $n$ and designed distance $\delta \geq 5$. Clearly the
$d_z=\wt(C_2) > d_x=\wt(C_1)=3$. Let $k_2$ be the dimension of
$C_2$, then one can derive asymmetric quantum code with parameters
$[[n,k_1+k_2-n,d_z/3]]_q$. In fact, one can short the columns of the
parity check matrix of the Hamming code $C_1$ to obtain a cyclic
code with less dimension and large minimum distance, in which it can
be used as $C_2$.
\end{exampleX}
\bigskip

\section{Conclusion and Discussion}
We presented two  generic methods to derive asymmetric quantum error
control codes based on two classical cyclic codes over finite fields. We showed that one
can always start by a cyclic code with arbitrary dimension and minimum distance,  and will be able to derive AQEC
using the CSS construction. The method is also used to derive a family of subsystem codes.

\medskip

Based on the generic methods that we develop, all classical cyclic codes can be used to construct asymmetric quantum cyclic codes and subsystem codes. In a quantum computer that utilizes asymmetric quantum cyclic codes to protection quantum information, such codes are superior in  a sense that online encoding and decoding circuits will be used. In addition quantum shirt registers can be implemented.  Our future will include bounds on the minimum distance and dimension of such codes. Furthermore such work will include the best optimal and perfect asymmetric quantum codes.

\medskip

Such asymmetric quantum error
control codes aim to correct the  phase-shift
errors that occur more frequently than qubit-flip errors. An attempt to address the fault tolerant operations and quantum circuits of such codes are given in~\cite{stephens07}, where an analysis for Becan-Shor asymmetric subsystem code is analyzed and a fault-tolerant circuit is given.

\smallskip

\emph{\begin{small}
S.~A.~A. dedicates this paper to Dr. Moustafa Mahmoud who passed away in 10/31/2009  at the age of 88. Dr. Mahmoud was an Egyptian scientist and a prolific author, who boarded the ship of natural science, medicine, physics, knowledge, philosophy, and religion. He authored books and presented more than 400 TV video lectures to deeply explain the earth, sun, time, life, death, space, Holy scriptures,  quantum theory and A. Einstein's work.
\end{small}}

\scriptsize
\bibliographystyle{plain}

\end{document}